\documentclass[10pt,conference]{IEEEtran}
\IEEEoverridecommandlockouts
\usepackage{amsmath}
\usepackage{amsthm}
\usepackage{amssymb}
\usepackage{setspace}
\usepackage{authblk}
\usepackage[noend,boxed]{algorithm2e}
\usepackage[table]{xcolor}
\usepackage{graphicx}

\newtheorem{theorem}{Theorem}[section]
\newtheorem{lemma}{Lemma}[section]
\newtheorem{proposition}{Proposition}[section]

\theoremstyle{definition}

\renewcommand{\IEEEQED}{\IEEEQEDopen}

\begin{document}
\title{A Secure Distributed Authentication scheme based on CRT-VSS and Trusted Computing in MANET}
%\title{Secure Distributed Authentication scheme using CRT-VSS and Trusted Computing in MANET}

\author[]{Qiwei Lu}
\author[]{Wenchao Huang}
\author[]{Xudong Gong}
\author[]{Xingfu Wang}
\author[]{Yan Xiong}
\author[]{Fuyou Miao}

\affil[]{ University of Science and Technology of China, Hefei, China
\authorcr \{Wangxfu, yxiong, huangwc, mfy\}@ustc.edu.cn  \authorcr \{luqiwei, lzgxd\}@mail.ustc.edu.cn }
\maketitle

\begin{abstract}

With the rapid development of MANET,  secure and practical authentication  is becoming  increasingly  important. The existing works perform the research from two aspects, i.e., (a)secure key division and distributed storage, (b)secure distributed authentication. But there still exist several unsolved problems. Specifically, it may suffer from  cheating problems and fault authentication attack, which can result in authentication failure and DoS attack towards authentication service. Besides, most existing schemes are not with satisfactory efficiency  due to exponential arithmetic based on Shamir's scheme.
In this paper,  we explore the property of verifiable secret sharing(VSS) schemes with Chinese Remainder Theorem (CRT), then propose a secret key distributed storage scheme based on CRT-VSS  and trusted computing for MANET. Specifically, we utilize trusted computing technology to  solve two existing cheating problems in secret sharing  area before. After that,  we do the analysis of  homomorphism property with CRT-VSS and  design the corresponding shares-product sharing scheme  with better concision.
On such basis, a secure distributed Elliptic Curve-Digital Signature Standard signature (ECC-DSS) authentication scheme based on CRT-VSS scheme and trusted computing is proposed.
 Furthermore, as an important property of authentication scheme, we discuss the refreshing property of CRT-VSS  and do thorough comparisons with Shamir's scheme.
Finally, we provide formal guarantees towards our schemes proposed in this paper.

\end{abstract}


%\IEEEpeerreviewmaketitle
%\enlargethispage{4\baselineskip}
\section{Introduction}
\label{Introduction}
Mobile ad hoc network(MANET) is a new kind of wireless mobile network, where  nodes are dispersed and easy to be compromised. So it isn't reasonable to assume such a trust center in MANET. This makes the central management of secret keys and  traditional  authentication scheme with single CA weak and dangerous. In order to solve such problem and build dependable schemes, most of the existing work contributes to the design of  secure distributed authentication scheme in such two steps as follows,
 \begin{itemize}
   \item  divide the secret key in some way (e.g., secret sharing scheme) and distribute them to different nodes.
   \item  design the detailed secure and dependable authentication process, whose essence is the reconstruction of  original secret key and generation of  authentication signature.
 \end{itemize}

Specifically, in the first step,  the private key of CA will be divided in shares by threshold cryptography scheme,e.g. $(t,n)$ threshold Shamir's scheme\cite{Shamir}. Then the shares will be stored into $n$ different nodes and each node has one.
As for the  second step, distributed  authentication, many schemes are proposed to provide the service of a distributed CA.
%Due to the distributed property of authentication process, it should be executed in a secure and dependable way in case of various unknown attacks.
And there have been several corresponding schemes, e.g. the multi-hop authentication scheme  with  RSA function in \cite{xiongyan}, ElGamal shape function scheme in \cite{Kaya08} and  threshold DSS (Digital Signature Standard) scheme in \cite{gennaro1996robust}.

However, there still exist several problems. The threshold secret sharing schemes all have two categories  of cheating problems against secret distributor and participants respectively. And existing authentication solutions can't eliminate the possibility of cheating (just detect it with some computation overhead). Most of existing distributed  authentication schemes, e.g., \cite{gennaro1996robust},\cite{xiongyan},  lack  the validation of the secret shares  and trusted property of the nodes, which will result in fault attack problem \cite{faultattack2010}. Furthermore, the misbehavior of authentication nodes may lead to DoS towards the authentication service. What's worse, most of existing works based on Shamir's scheme\cite{Shamir} suffer from  a relative high complexity of exponent arithmetic in secret reconstruction phase.
%Then upper layer authentication scheme will suffer low efficiency performance and  prone to  DoS result in attack environment.  Some discussion on the limitation and restriction of the existing schemes can  be seen in   \cite{ertaul2005security},\cite{xiongyan}, etc.

Due to the reasons above, we aim to contribute as below:

\begin{itemize}
  \item  First, we propose a new scheme  based on CRT-VSS (Verifiable Secret Sharing Scheme Based on  Chinese Remainder Theorem) scheme instead of Shamir's scheme, tackling  security problems with better efficiency. Specifically, we utilize the trusted computing technology to solve two cheating problems above.

  \item  Then  we explore $(\oplus,\oplus)$, $(\otimes,\otimes)$ homomorphism property for  CRT-VSS scheme  and  design the secure shares-product sharing scheme with CRT scheme. It owns better concision and equal security property compared with Shamir's scheme.

  \item On such basis, a trusted ECC-DSS distributed authentication scheme is proposed to eliminate the possibility of the DoS and fault attack problems  before.
       Furthermore, we discuss the refreshing property of CRT-VSS scheme and do thorough comparisons with Shamir's scheme  as an important property of authentication scheme towards mobile adversaries \cite{LidongZhou} tolerance
 \item       Finally, we  model our schemes and ensure the security with formal guarantees.

\end{itemize}

The rest of this paper is organized as follows. In Section~\ref{related  work}, we introduce related works  in this paper. In Section~\ref{simple authentication}, we design and analyze the framework of authentication  between two nodes(as initiator and responder) and prepare it as  a module which will be used in detailed scheme below. Section~\ref{CRT-VSS} shows the secret key distributed storage scheme based on CRT-VSS  and trusted computing.  In section~\ref{homomorphic property} we explore and do some analysis for the homomorphism property of the CRT-VSS scheme. Then we propose our trusted ECC-DSS distributed authentication scheme in section~\ref{ECC-DSS authentication}.
 %In Section~\ref{refreshing scheme}, we discuss the refreshing scheme for CRT-VSS scheme and do thorough comparisons with Shamir's scheme.
  In Section~\ref{security analysis}, we  present security analysis on the schemes in this paper. Finally, Section~\ref{Conclusion} concludes the paper. % and mentions some future work.

\section{related  work}
\label{related  work}

\subsection{ threshold secret sharing scheme}
In traditional $(t,n)$ threshold secret sharing schemes, the secret distributor $D$ divides the secret $S$ into $n$ shares and deliver them to $n$ participants securely. The  recovery of secret $S$ can  only be accomplished by a coalition  not less than $t$ participants. The  most popular secret sharing scheme is  Shamir's scheme in \cite{Shamir}. The reconstruction of  secret can be realized  with complexity $O(tlog^2t)$ \cite{Shamir}.
The first secret sharing scheme based on the Chinese Remainder Theorem(CRT) is proposed in \cite{CRT} with reconstruction complexity $O(t)$.
There exist two kinds of cheating problems in secret sharing area. (1) cheating from distributor. Distributor delivers the false secret shares to corresponding authentication nodes, braking the base of authentication. (2) cheating among  participants. During the reconstruction phase, malicious participant will provide false/fault shares, leading to unfair, inconsistent results, etc.

The verifiable secret scheme development based on zero-knowledge proof evidence has both improvements on the original Shamir's and CRT scheme(e.g. \cite{Iftene},\cite{lixiong}), which can only be used to detect such cheating problems but without elimination.  The detection  scheme in \cite{harn2009detection} can achieve the goal in some specified condition, but inefficient and restricted to specified cheater number condition.  Though work in \cite{CRTRange} uses the range proof techniques to prove the right range of the  generated parameters, it can't provide further assurance or evidence about the right property of the node.

\subsection{ homomorphism property and distributed authentication}
In order to tolerate mobile adversaries, the key management service proposed in \cite{LidongZhou} employs share refreshing as an important scheme which relying on the $(\oplus,\oplus)$ homomorphism property. It is indicated in \cite{homomorphisms} that the traditional Shamir's threshold scheme only has the $(\oplus,\oplus)$-homomorphism property, without the $(\otimes,\otimes)$ or $(\oplus,\otimes)$- property. While it's only mentioned in \cite{homomorphisms},\cite{Iftene} that the extended Asmuth-Bloom scheme has homomorphism proprieties but without any proof about it. The discussion about  homomorphism property of  CRT-VSS scheme is rare.

In the Distributed CA scheme in \cite{LidongZhou}, the private key of CA is divided into several shares and distributed to several chosen entities, and each has one and only one part. Distributed authentication is the result of the Distributed CA scheme in the distributed environment of MANET. The scheme proposed in \cite{xiongyan} realize the authentication with the process of the multi-hop signing. The encrypted  functions mechanism\cite{Encryptedfunction} is used to strengthen the privacy and security of the authentication process, similar related work can be seen in \cite{ertaul2005security}. In \cite{gennaro1996robust}, a meaningful robust DSS signature schemes are proposed under  different security requirements.
But in all the schemes proposed above, the choice of nodes coalition for authentication is random and  lacks  validation, including the secret shares and the trusted property of the nodes. This may lead to DoS  and fault attack problem  \cite{faultattack2010} . Besides, RSA/Elgamal  encrypt/decrypt  and signature mechanism  utilized by all the schemes above has a relative limited efficiency in MANET according to analysis in \cite{ertaul2005security}.

\subsection{trusted computing}
%\subsubsection{CRT-VSS scheme}
Trusted Computing(TC) is a new technology developed and promoted by the Trusted Computing Group(TCG) for enhancing the security of computers and networks. The root of trust is a tamper resistant hardware engine, called TPM. TPM is assumed robust against both hardware and software attacks from either the underlying host or external sources. The aim of TC is to allow someone else to verify that only authorized code runs on a system. There are at least 16 PCRs in a TPM and a PCR is a 160-bit shielded storage location to hold an unlimited number of measurements in the way like this: $PCR_{new}$=SHA-1$(PCR_{old}\mid \mid measurement)$. A simple way to accomplish the measurement mentioned here is to do a digest of the configuration of the platform to ensure the integrity of a Platform.
Direct anonymous attestation (DAA) protocol \cite{DAA} was adopted by TCG as the method for remote authentication of the hardware module TPM, while preserving the privacy of the user of the platform that contains the module.  Furthermore, the Property-Based Attestation(PBA) scheme proposed in \cite{PBA} present a new privacy-preserving approach to attest the right configuration of host.
Now the application of TC technology appears in various areas, e.g., inference Control \cite{IC}, Ad hoc routing protocol security \cite{DAAODV}, digital rights management, etc.

\section{framework of authentication protocol based on trusted computing between two nodes}
\label{simple authentication}
In the two important steps of a secure authentication scheme, i.e. secure secret division/distribution and secure distributed authentication, the right behavior of the nodes is so important that it should be ensured in some reliable authentication way. Otherwise the malicious action of nodes may lead to cheating problems and DoS/fault attack.  In order to solve this problem, we will utilize the trusted computing technology, including the property of TPM, the secure and authentication property of trusted computing protocols, e.g. DAA\cite{DAA} and PBA\cite{PBA} protocol. Because the authentication between two nodes will be referred and used repeatedly in the detailed schemes below, here we will first illustrate the framework of the authentication protocol between two nodes, one initiator and responder, in the MANET. Then the other part of this paper can refer it as a module easily. Without loss of generality, suppose the authentication process happens between node $P_i$ and $P_j$, node $P_i$  request for authentication from  node $P_j$. The  authentication process between them is as follows in Figure~\ref{fig:frameworkAUTH}.
\begin{figure}[h]
 \centering
 \fbox{
\begin{minipage}{3.2 in}
1. $P_i$ $\rightarrow$ $P_j$: $id_i$,$id_j$,INIT\\
2. $P_j$ $\rightarrow$ $P_i$: $id_j$,$id_i$,$n_1$\\
3. $P_i$ $\rightarrow$ $P_j$: $id_i$,$id_j$,
                            $im$=SHA-1$(id_i||n_1||PCR)$,\par
\hspace{6ex}                             $DS_i=DAASign(im)$,$PS_i=PBASign(im)$,\par
\hspace{6ex}                             SHA-1$(id_i||$K\_INFO\_I$||im||DS_i||PS_i)$,K\_INFO\_I \par
4. $P_j$ $\rightarrow$ $P_i$: $id_j$, $id_i$,$S_{sk_{AIK_J}}$(K\_INFO\_J$||id_j$),$\{id_j||n_{2}\}_{k_{ij}}$ \par
5. $P_i$ $\rightarrow$ $P_j$: $id_i$,$id_j$,$\{id_i||n_2\}_{k_{ij}}$ \par
\end{minipage}
}
 \caption{$P_i$ request for authentication from node $P_j$ }
 \label{fig:frameworkAUTH}
\end{figure}

%In our previous work, similar authentication process can be seen in \cite{DAAODV}, where  authentication process is designed to enhance the security of the route protocol against the attacks, e.g. tunnel attack, vertex attack  and malicious attacks. And similar work can be seen in \cite{IC}, where the authentication happens between the ICM and ACM party to ensure the trusted authentication happens before the response of the query and IC\_POLICY execution. But the concrete condition and purpose of the authentication process and the details are different.

In the framework, each node maintains three lists:
 \begin{itemize}
   \item Neighbor list $N$. records the IDs of the neighbor nodes.
   \item  Trust list $T$.  records the authentication result of different neighbors with different symbols, e.g. 1 represents trusted, 0 for no authentication before and -1 for authentication failure, initially all 0.
   \item Key list $K$. stores the established session with trusted neighbor nodes, initially empty.
 \end{itemize}

Now we will do some explanation to the details of the framework proposed above as follows.

 \fbox{
\begin{minipage}{3.2 in}
1. $P_i$ $\rightarrow$ $P_j$: $id_i$,$id_j$,INIT
\end{minipage}
}

Node $P_i$ sends the request for authentication towards node $P_j$ with its ID $id_i$ and destination ID $id_j$.

 \fbox{
\begin{minipage}{3.2 in}
2. $P_j$ $\rightarrow$ $P_i$: $id_j$,$id_i$,$n_1$
\end{minipage}
}

After receiving the request from node $P_i$, $P_j$ first check the  ID $id_i$ in neighbor list $N$, if there is no match then drop the request without any reply, else node $P_j$ reply with  its ID $id_j$ and a random nonce $n_1$ as the challenge.

\fbox{
\begin{minipage}{3.2 in}
3. $P_i$ $\rightarrow$ $P_j$: $id_i$,$id_j$,
                            $im$=SHA-1$(id_i||n_1||PCR)$,\par
\hspace{6ex}                             $DS_i=DAASign(im)$,$PS_i=PBASign(im)$,\par
\hspace{6ex}                             SHA-1$(id_i||$K\_INFO\_I$||im||DS_i||PS_i)$,K\_INFO\_I \par
\end{minipage}
}

After receiving the $n_1$ from $P_j$, node $P_i$ will do the following work.
\begin{itemize}
  \item Compute the measurement of integrity towards the platform by $im$=SHA-1$(id_i||n_1||PCR)$. Here we note that the measurement of the integrity consists of the content of PCR. Due to the TPM will record the measurement of the paltform from time to time by the function and update of PCRs. Here we use the command TPM\_Quote to get the content of the specified PCR. And the PCR can be updated by the command TPM\_Extend.
  \item Generate the DAA signature to attest to $P_j$ its perfect TPM. Here the the integrity measurement $im$ computed above will be the  parameter,$DS_i=DAASign(im)$.
  \item Generate the PBA signature to furthermore attest to $P_j$ that its configuration and property is in  a good set but without revealing the details. It is assumed that before running the PBA protocol, node $P_i$ and $P_j$ have already agreed on a set of configuration values denoted $ CS = \{cs_1,\dots,cs_n \}$.
            \footnote{We notice that, the configuration values in $CS$ may vary from node to node and time to time, to synchronize with updates and to cover all nodes' configuration is important but not the focus of this paper, we assume such an agreement by some specific  protocol. }

The generation denoted as $PS_i=PBASign(CS,im)$. Note that $CS$ is constructed and planted inside the nodes initially and will be updated consistently if necessary according to the technology of trust computing.
  \item If there is no key established with node $P_j$, $K[id_j]$ is empty, node $P_i$ will generate its public part parameters according to Diffie-Hellman key exchange scheme, denoted as  K\_INFO\_I=($g$,$p$,$g^x\ mod \ p$) and  provide the integrity assurance by SHA-1$(id_i||$K\_INFO\_I$||im||DS_i||PS_i)$, otherwise  the digest assurance is SHA-1$(id_i||im||DS_i||PS_i)$
\end{itemize}

\fbox{
\begin{minipage}{3.2 in}
4. $P_j$ $\rightarrow$ $P_i$: $id_j$, $id_i$,$S_{sk_{AIK_J}}$(K\_INFO\_J$||id_j$),$\{id_j||n_{2}\}_{k_{ij}}$
\end{minipage}
}

After receiving the signature attest message from $P_i$, node $P_j$ will do the following work.
\begin{itemize}
  \item Check the integrity of the message by digest assurance above. If check fails, then drop the message and abort.
  \item Validate the  DAA and PBA signature by $DAAVerify(DS_I)$,$PBAVerify(CS,im)$.
  If any verification fails, abort the protocol and mark the trusted  symbol  by $T[id_i]=-1$, otherwise $T[id_i]=1$.
  \item  Like above, if there is no key established before, then generate its public part parameters  as  K\_INFO\_J=($g$,$p$,$g^y\ mod \ p$). Furthermore sign the K\_INFO\_J structure with its ID under the private key of AIK inside the TPM to attest the right origin, in order to avoid the possibility of man-in-the-middle attack. Then the session key between the nodes can be generated  according to DH scheme, denoted as $k_{ij}=(g^x)^y \ mod\ p$. Then node $P_j$ stores the session key $k_{ij}$ in the key list as $K[id_i]=k_{ij}$  and  generate the encryption part $\{id_j||n_2\}_{k_{ij}}$ with its ID $id_j$ and random nonce $n_2$ inside.
\end{itemize}

\fbox{
\begin{minipage}{3.2 in}
5. $P_i$ $\rightarrow$ $P_j$: $id_i$,$id_j$,$\{id_i||n_2\}_{k_{ij}}$ \par
\end{minipage}
}

\begin{itemize}
\item $P_i$ first validate the signature with the public AIK of $P_j$ by $S_{pk_{AIK_J}}(S_{sk_{AIK_J}}$(K\_INFO\_J$||id_j$)) $\stackrel{?}{=}$ K\_INFO\_J'$||id_j$. If it fails, abort and  without any reply.
 \item Similarly, if there is no key established before, $P_i$ first generate the session key $k_{ij}$ with the K\_INFO\_J and its own part K\_INFO\_I according to DH scheme denoted as $k_{ij}=(g^y)^x \ mod\ p$. Then stores the key $k_{ij}$ in its key list.
 \item Then node $P_i$ decrypt the encrypted message $\{\{id_j||n_2\}_{k_{ij}}\}_{k_{ij}^{-1}}$ and get the nonce $n_2$.
 \item At last node $P_i$ generates the encryption part $\{id_i||n_2\}_{k_{ij}}$ with its ID $id_i$ and nonce $n_2$ above together to inform neighbor  node $P_j$ that the key $k_{ij}$ has been  well established between them.
\end{itemize}

Some additional explanation about the details of the scheme above is provided as follows.
\begin{enumerate}
  \item The Neighbor list is constructed by routing protocols, e.g., OLSR, AODV. Besides, it and can consist of the nodes not neighbored stringently, e.g., the two-hop neighbors in OLSR, in order to contain enough neighbors( not less  than $n$).
  \item The update and change of Neighbor list $N$ will result in the update of the trust list $T$ and key list $K$. Once neighbor node moves out of the region, the item will be removed out from the neighbor list and that will lead to corresponding item abolishment in the trust list and key list respectively.
  \item In order to enhance the security, the trust list $T$ and key list $K$ should be re-initialized periodically. So the authentication process should be executed periodically when necessary.
\end{enumerate}

In the other parts of the paper, we will refer this authentication process as a module denoted as $auth(id_i,id_j)$ indicating node $P_i$ running authentication process with node $P_j$ and  request for the authentication from the later. The security analysis about the protocol will be left in section~\ref{security analysis}.

\section{ secret key distributed storage scheme based on CRT-VSS and  trusted computing}
\label{CRT-VSS}
%The secret key distributed storage scheme based on  CRT-VSS and  trusted computing  proposed below  is based on the Iftene's  scheme in \cite{Iftene}. As mentioned above, Iftene's scheme in \cite{Iftene} is verifiable and it can detect the cheating of the participants. But it cannot eliminate the unfair result and the cheating of the distributor.

%Before we show the details of our scheme, we summarized some notations used in the paper in Table~\ref{tab:crtnotations}.
%
%\begin{table}
%\centering
%\caption{Notations }
%\label{tab:crtnotations}
%\begin{tabular}{|l|l|}
%\hline
%Notation & Explanation  \\
%\hline
%$S$ & the secret key of CA.\\
%$D$ & the distributor of the secret S. \\
%$n$ & the number of the shares.\\% distributed in nodes \\
%$t$ & the threshold or the minimum number of \\
% &  participants to reconstruct the secret S.\\
%$P_i:1 \leq i \leq n$ & the participants of the reconstruction. \\
%$p_i:1 \leq i \leq n$ & n safe primes.  \\
%$P$ & the result of $\prod_{i=1}^{n}p_i$.\\
%$y$ & $d+Am_0$,where A is a random number.\\
%$y_i:1 \leq i \leq n$ &  y mod $m_i$,the share of node i.\\
%$\mathbb{P}$ & the set of size $n$ participants.\\
%
%$\mathbb{C}$ & a coalition of size $t$ participants.\\
%$M_C$ & the modulus of coalition $C$, $\prod_{i \in C}{m_i}$\\
%$M$ & the product of $m_i$, $\prod_{1 \leq i \leq n }{m_i}$,\\
% & also as the domain of the $y \in Z_M$.\\
%\hline
%\end{tabular}
%\end{table}

 Our scheme has two phases. One is  the distributing phase in which distributor divides the secret $S$ into $n$ shares(denoted as $\{S_i\}_{i \in \{1,\dots,n\}}$) and distributes them to a set of $n$ nodes as participants(denoted as set $\mathbb{P}$). The other is reconstruction phase, any coalition of $t$ nodes(described  as $\mathbb{C}$) can reconstruct the secret $S$ together. The details can be described as follows.

\subsection{Distribution phase }
 To share a secret $S$ among a set of $n$ nodes with verifiable shares, the distributor does the following:\\
1) Setup work. Distributor $D$ setup and set the related parameters, including the big prime number $m_0(m_0>S)$, the positive sequence $m=\{m_1,m_2,\dots,m_n\}$ meets the requirements below:
     \begin{itemize}
     \item the sequence $m=\{m_1,m_2,\dots,m_n\}$ is increasing;
     \item $(m_i,m_j)=1,(i \neq j);$
     \item $(m_i,m_0)=1,(i=1,2,\dots,n)$;
     \item $N=\prod_{i=1}^{t}m_i > m_0^{2}\prod_{i=1}^{t-1}m_{n-i+1};$
     \item ensure the $p_i=2m_i+1$ is still prime.
     \end{itemize}
     After that  distributor need to choose a random positive integer $A$ satisfied the condition $0<A<M$,and $y=S+Am_0<M$, here the $M=\prod_{i=1}^{t}m_i$. Then $y$ is the real shape of the secret $S$ to be divided later.\\
2) The distributor $D$ choose $n$ trusted neighbors nodes $\{P_i\}_{i \in \{1,\dots,n\}}$(denoted as set $\mathbb{P}$) to be qualified participants to distribute the shares later.
     \begin{itemize}
       \item D $\stackrel{B}\rightarrow$ $\{\mathbb{N_D}\}$: $id_D$,AUTH,$S_{sk_{AIK}}(id_D)$. $\{\mathbb{N_D}\}$ is the neighbor set of D.
       The distributor $D$ broadcasts(denoted as $\stackrel{B}\rightarrow$ notation) a AUTH message to inform its neighbors the authentication election event. The message consists of  the ID of D, authentication symbol AUTH and the signature of the $id_D$ by the private key of the AIK inside the TPM of distributor D to attest to $P_i \in \{\mathbb{N_D}\}$ the identify of D.
       \item $P_i$ $\rightarrow$ D: $id_i$,$id_D$,$auth(id_i,id_D)$.\par
       After receiving the AUTH message, $P_i$ first check the signature with the public AIK of D by $S_{pk_{AIK_D}}(S_{sk_{AIK_D}}(id_i)) \stackrel{?}{=} id_i$. If it holds, then $P_i$ begins to execute the authentication process as $auth(id_i,id_D)$, else $P_i$ drop the AUTH message without any reply.
       \item After the authentication process, Distributor D choose $n$ qualified nodes from the trust list $T$ as participants.
     \end{itemize}
%The distributor first using the DAA scheme\cite{DAA} and PBA challenge protocol\cite{PBA} to check the trusted property and right configuration of the nodes, the DAA scheme can be used to ensure the existence of a trusted and not compromise TPM in the node, the PBA challenge protocol is used to ensure the trusted property of the chosen nodes.
3)  Reversely, the distributor D executes the authentication process to attest to the  selected participants $\{P_i\}_{i \in \{1,\dots,n\}}$.
 \begin{itemize}
   \item D $\rightarrow$ $\{P_i\}$($P_i \in \mathbb{P}$): $auth(id_D,id_i)$. \par
   \item $P_i$ $\rightarrow$ D: $id_i$,$id_D$, AGREE.\par
   If $P_i$ pass the authentication process from D and agree to undertake the responsibility as a participant, $P_i$ reply with a AGREE message to inform the distributor D, else deny and refuse the request without reply.
 \end{itemize}

 If D doesn't receive the AGREE message from $P_i$, D have to remove the id of the $P_i$ from the participant set $\mathbb{P}$. This step won't  finish until the participant set has $n$ participants, equally receiving $n$ AGREE messages totally.\\
% pant also use DAA and PBA scheme to validate the configuration of $D$ in order to ensure the property of $D$ is trusted to avoid the first category cheating problem mentioned in section II.
4)  The distributor $D$ generates and distributs shares and validation information securely.
 \begin{itemize}
   \item D $\rightarrow$ $\{P_i\}$($P_i \in \mathbb{P}$): $id_D$,$id_i$, $\{S_i\}_{k_{iD}}$. \par
   Let $g_i \in Z_{p_i}$ be an element of order $m_i$. D computes the shares by $S_i=y (mod \ m_i)$ and encrypts $S_i$ by the session key established in the authentication process before. This will keep the secret of the shares.
   \item D $\stackrel{B}\rightarrow$ $\{P_i\}$: $id_D$,PUBLIC,$S_{sk_{AIK_D}}(id_D||p_i||g_i||m_i)$. \par
  Then distributor generates the validation information by $ z_i = g_i^{S_i} (mod \ p_i)$ and  makes the values $p_i$,$ g_i$,$ z_i$ public by the broadcast PUBLIC message. In order to ensure the correctness of  the origin of the public information, the signature with private AIK $sk_{AIK_D}$ is included in the message.
 \end{itemize}

After receiving the PUBLIC message, the participants can verify the correctness of the origin of  public information by the signature. With the correct validation information, participant $P_i$ can find whether his share is valid or not by checking $ z_i = g_i^{S_i} (mod \ p_i)$.

The work and schemes proposed before all assume that $D$ is honest. But once the malicious distributor $D$ distributes the shares for some $y>M$, then $y$ will be greater than $M_C$ for all coalition $C$ of size $t$. Hence, according to the analysis in \cite{CRTRange} , $C$ may not compute the correct $y$ value or recover correct secret $S$ with inconsistent results. But here in the scheme, with the authentication help of the trusted computing technology, the first cheating problem exists in this phase before can be eliminated.

\subsection{reconstruction  phase }
As can be seen in the notations Table~I, without loss of generality, suppose
$C={\{P_{i_k}\}}_{k \in \{1,\dots,t\}}$,$M_C=\prod_{i\in C}{m_i} $, $M_{C \backslash \{i_k\}}=\prod_{j\in C,j \neq i_k}{m_j} $, $M_{C,i_k}^{'}\equiv M_{C \backslash \{i_k\}}^{-1}(mod \ m_{i_k})$, $M_{C,i_k}^{'}M_{C \backslash \{i_k\}}^{-1}\equiv 1(mod \ m_{i_k})$,  the exchange of shares and reconstruction happens between the participants in the coalition $C$. $P_{i_k},P_{i_j} \in C$. The details are as follows.\\
1)$P_{i_j}$ should firstly ensure that the other  participants, denoted as  $P_{i_k}$$(i_k \neq i_j)$,in the coalition $C$ are all trusted.
\begin{itemize}
  \item If $T[i_k]=1$ and $K[i_k]$ not empty, then $P_{i_k}$ is trusted and goto the next step.
  \item If $T[i_k]=-1$, then the authentication from $P_{i_k}$  before failed, $P_{i_j}$ abort and the coalition should be rebuilt.
   \item If $T[i_k]=0$, no authentication before. $P_{i_j}$ asks for new $auth(id_{i_k},id_{i_j})$ process for the $P_{i_k}$ and deal with the result as before. If the authentication succeeds, goto the next step.
\end{itemize}
2) After the insurance towards the participants in coalition $C$, $P_{i_j}$ delivers its share $S_{i_j}$ to other participant securely.\par
$P_{i_j}$ $\rightarrow$ $P_{i_k}$$(i_k \neq i_j)$: $id_{i_j}$,$id_{i_k}$,$\{S_i\}_{k_{i_ji_k}}$.

 The other participants can decrypt the message with the session key $k_{i_ji_k}$, then  validate the shares from other nodes, e.g. The share $S_{i_j}$ can be verified by the other participants  in $C$ by the equation
    $ z_{i_j} \stackrel{?}{=} g_{i_j}^{S_{i_j}} (mod \ p_{i_j})$.\\
3) Similarly, $P_{i_j}$  can receive the encrypted shares from the other $(t-1)$ participants, after the decryption and validation success, $P_{i_j}$  can get the equation set established as below:
$$  S_{i_j}\equiv S (mod \ m_{i_j}), (i=1,2,\dots,t)$$
% \begin{equation*}
%\left\{
%\begin{array}{c@{\;\;}l}
%y \equiv S_{i_1}(mod \ m_{i_1})\\
%y \equiv S_{i_2}(mod \ m_{i_2})\\
%\dots \dots \dots\dots \dots \\
%y \equiv S_{i_t}(mod \ m_{i_t})
%\end{array}
%\right.
%\end{equation*}
According to the Chinese remainder  theorem, the only solution can be easily computed as
\begin{equation}
\label{reconstruct}
y=\sum_{k=1}^{t}S_{i_k}M_{C,i_k}^{'} M_{C \backslash \{i_k\}}\quad (mod \ M_C)
\end{equation}

So $S=y(mod \ m_0)$. the  reconstruction can be accomplished and all secret $S \in Z_{m_0}$ can be reconstructed.
The work and schemes proposed before all assume that the exchange of the encrypted shares are all correct and without malicious forgery. But this can't be ensured when the participants are corrupted and act as a malicious attacker. But in the scheme here,  the authentication process  ensures that the second cheating problem exists in this phase can be eliminated too.

%Actually, the reconstruction can be done by a party who collects $t$ shares, This will be better fit for the environment of  service, e.g. authentication, signature, secure message transmission, etc.

In other parts of the paper, we will use a concision  form $S  \xrightarrow{CRT} \{S_i\}_{i \in \{1,\dots,n\}}$ to indicate  secret $S$'s  division into $n$ shares $\{S_i\}_{i \in \{1,\dots,n\}}$ according to CRT-VSS scheme with prime modulus parameters $\{m_i\}_{i \in \{1,\dots,n\}}$.
Similarly, $S=CRT_I(\{S_{i_k}\})_{k \in \{1,\dots,t\}}$ represents  secret  reconstructed by coalition $I=\{i_k\}_{k \in \{1,\dots,t\}}$, $|I|=t$  as Equation~\ref{reconstruct}.

\section{homomorphism property for CRT-VSS scheme, basis and preliminaries}
\label{homomorphic property}
In this section, we will explore the homomorphism property of CRT-VSS scheme and do some  preliminaries for the trusted ECC-DSS distributed authentication scheme later. It's only mentioned in \cite{homomorphisms},\cite{Iftene} that the extended Asmuth-Bloom scheme has homomorphic proprieties but without any proof about it and such discussion about the homomorphic property of CRT-VSS scheme is rare.

%The literature before all designs the share refreshing scheme with the Shamir's threshold scheme, such work can be seen in \cite{LidongZhou},\cite{xiongyan}. They all set the refreshing shares 0, then generate the refreshing subshares.  Here we will extend such refreshing schemes based on  $(\oplus,\oplus)$-homomorphism property, then build the corresponding refreshing idea with the  $(\oplus,\oplus)$-homomorphism of CRT scheme.

%\begin{definition}(the homomorphism property\cite{homomorphisms})\label{homo}
%Let $\oplus$ and $\otimes$ be binary  functions on the elements of the secret domain $S$ and of the share domain $\{S_i\}_{i \in \{1,\dots,n\}}$. A $(t,n)$ threshold scheme has the $(\oplus,\oplus)$(or $(\oplus,\otimes)$, $(\otimes,\otimes)$)-homomorphism property if for all coalition $I=\{i_k\}_{k \in \{1,\dots,t\}}$, $|I|=t$, $\Re$ donates some secret reconstruction scheme, whenever
%$S_1=\Re_I(\{S_{i_k}^{1}\})_{k \in \{1,\dots,t\}}$ and  $S_2=\Re_I(\{S_{i_k}^{2}\})_{k \in \{1,\dots,t\}}$, then
%\begin{itemize}
%  \item $(\oplus,\oplus)$-property. $S_1 \oplus S_2=\Re_I(\{S_{i_k}^{1} \oplus S_{i_k}^{2}\})_{k \in \{1,\dots,t\}}$.
%  \item $(\oplus,\otimes)$-property. $S_1 \oplus S_2=\Re_I(\{S_{i_k}^{1} \otimes S_{i_k}^{2}\})_{k \in \{1,\dots,t\}}$.
%  \item $(\otimes,\otimes)$-property. $S_1 \otimes S_2=\Re_I(\{S_{i_k}^{1} \otimes S_{i_k}^{2}\})_{k \in \{1,\dots,t\}}$.
%\end{itemize}
%
%\end{definition}

\begin{lemma}
\label{CRTproperty}
For any coalition $I$, $|I|=t$, the notation $CRT_I$ donates the reconstruction scheme, then all CRT-VSS scheme has  $(\oplus,\oplus)$ and $(\otimes,\otimes)$-homomorphism property, whenever
$S_1=CRT_I(\{S_{i_k}^{1}\})_{k \in \{1,\dots,t\}}$, $S_2=CRT_I(\{S_{i_k}^{2}\})_{k \in \{1,\dots,t\}}$, then
\begin{itemize}
  \item $(\oplus,\oplus)$. $S_1 \oplus S_2=CRT_I(\{S_{i_k}^{1} \oplus S_{i_k}^{2}\})_{k \in \{1,\dots,t\}}$.
  \item $(\otimes,\otimes)$. $S_1 \oplus S_2=CRT_I(\{S_{i_k}^{1} \otimes S_{i_k}^{2}\})_{k \in \{1,\dots,t\}}$.
\end{itemize}
\end{lemma}
%\begin{proof}
%(a) For $(\oplus,\oplus)$-property. The proof is simple and easy similar to the proof on Shamir's scheme in \cite{homomorphisms}.
%(b) For $(\otimes,\otimes)$-property. The proof  can be turned into the modulus equation set below:
% $S_1\times S_2 \stackrel{?}{\equiv} S_{i_{k}}^1\times S_{i_{k}}^2(mod \ m_{i_{k}})$, $k \in \{1,\dots,t\}$
%  Since $S_i^1\times S_i^2(mod \ m_{i_k})
%                 \equiv (S_1\ mod \ m_{i_k})\times(S_2\ mod\ m_{i_k})
%                  \equiv S_1^1 \times  S_1^2 (mod\ m_{i_k})$.
%     So the equation set above holds. Thus, $S_1 \otimes S_2=CRT_I(\{S_{i_k}^{1} \otimes S_{i_k}^{2}\}_{k \in \{1,\dots,t\}}$ holds, the CRT-VSS scheme has such  $(\otimes ,\otimes )$-homomorphism property.
%\end{proof}

\begin{proof}
The proof is simple and easy similar to the proof on Shamir's scheme in \cite{homomorphisms} and we neglect it here.
\end{proof}

Among the technical difficulties for designing distributed authentication schemes, the secure shares-product sharing scheme combing shares of two secrets, e.g. $a$ and $b$, into shares of the product of $a\cdot b$, e.g. $\{a_i b_i\}_{i \in \{1,\dots,n\}}$ is outstanding and important. Though related work and discussion can be seen in  \cite{BGW88},\cite{gennaro1996robust}, but they are all designed with Shamir's polynomial VSS scheme. Here we will first do some description about them, then design the corresponding scheme for CRT-VSS scheme.

\begin{lemma}
\label{product-Shamir}
The scheme combing shares of two secrets, e.g. $a,b \in Z$, into shares of the product of $a\cdot b$ with Shamir's polynomial VSS scheme is as follows.
\begin{enumerate}
  \item  Setup work. With $(t,n)$ Shamir's VSS scheme, shares $\{a_i\}_{i \in \{1,\dots,n\}}$ and $\{b_i\}_{i \in \{1,\dots,n\}}$ are generated by random polynomial denoted as $f_a$ and $f_b$ with secret $a$ and $b$ respectively.\\
  e.g. $a \stackrel{f_a}{\longrightarrow}\{a_i\}_{i \in \{1,\dots,n\}}$, $b \stackrel{f_b}{\longrightarrow}\{b_i\}_{i \in \{1,\dots,n\}}$,
  \item  Refreshing shares, $\{c_i\}_{i \in \{1,\dots,n\}}$ are generated by random  polynomial $f_c$ with secret 0,
  e.g.
      $0 \stackrel{f_c}{\longrightarrow}\{c_i\}_{i \in \{1,\dots,n\}}$
  \item  After participant node $P_i$  exchange  $v_i=a_ib_i+c_i$, each Participant node $P_i$ locally computes
        $ab=F_I(\{v_i\})_{i \in \{1,\dots,t\}}$.
\end{enumerate}
the scheme is $(t-1)$  eavesdropping adversary tolerant secure under the $(t,n)$ threshold condition.
\end{lemma}

The details about the correctness and security analysis can be seen in \cite{BGW88}. The  scheme consists of two main steps, degree reduction step ensure the correctness and randomization step strengthen the security. Since  $0 \stackrel{f_c}{\longrightarrow}\{c_i\}_{i \in \{1,\dots,n\}}$,  the reconstruction is  without retortion.
Furthermore, in order to maintain the $(t-1)$  eavesdropping adversary tolerant secure property(against the decompose attack towards Shamir's polynomial), the additional permutation term $c_i$ is necessary to keep the Shamir's polynomial secret and irreducible.
However, the weakness of the Shamir's scheme doesn't exist in CRT scheme and there is no need to keep such permutation term. Thus we can design the corresponding scheme with better concision and equal security.

\begin{theorem}
\label{product-CRT}
The scheme combing shares of two secrets, e.g. $a$ and $b$, into shares of the product of $a\cdot b$ with CRT scheme is $(t-1)$ eavesdropping adversary tolerant as follows.
\begin{enumerate}
  \item  Setup work. CRT scheme parameters $\{m_i\}_{i \in \{1,\dots,n\}}$, $M=\prod_{1\leq i \leq n}{m_i}$, $a,b \in Z_M$,$ab \in Z_M$. With CRT scheme,\\
  e.g. $a \stackrel{CRT}{\longrightarrow}\{a_i\}_{i \in \{1,\dots,n\}}$, $b \stackrel{CRT}{\longrightarrow}\{b_i\}_{i \in \{1,\dots,n\}}$,
  \item  After participant node $P_i$  exchange  $v_i=a_ib_i(mod\ n)$, each Participant node $P_i$ locally computes
        $ab=CRT_I\{v_i\}_{i \in \{1,\dots,t\}}$.
\end{enumerate}

the scheme is $(t-1)$  eavesdropping adversary tolerant secure under the $(t,n)$ threshold condition.
\end{theorem}
\begin{proof}
We will prove the correctness and security respectively.
\begin{enumerate}
    \item \textbf{Correctness.} According to  Lemma~\ref{CRTproperty},
 \begin{spacing}{0.6}
   \begin{small}
   \begin{align*}
    \begin{split}
  CRT_I&(\{v_i\})_{i \in \{1,\dots,t\}}=CRT_I(\{a_ib_i\})_{i \in \{1,\dots,t\}}\\
  &=CRT_I(\{a_i\})_{i \in \{1,\dots,t\}}\cdot CRT_I(\{b_i\})_{i \in \{1,\dots,t\}}\\
  &=a\cdot b(mod \ M)
    \end{split}&
    \end{align*}
    \end{small}
    \end{spacing}
So the reconstruction of $ab$ maintains well.

    \item \textbf{Security.}  On one hand, the reconstruction of $ab$ is $(t-1)$ tolerant, on the other hand, though the scheme is designed without such refreshing permutation term $c_i$, while the weakness of decomposing for  Shamir's scheme doesn't exist in CRT scheme. So the adversary can't get more information like the  decomposing attack with Shamir's scheme. Hence, the scheme is $(t-1)$  eavesdropping adversary tolerant secure.
\end{enumerate}
\end{proof}

According to  Theorem~\ref{product-CRT}, the secure shares-product sharing scheme based on CRT scheme has better concision(without the refreshing permutation term $c_i$) compared with Shamir's scheme in Lemma~\ref{product-Shamir}. From Lemma~\ref{CRTproperty} and Theorem~\ref{product-CRT}.

\section{a distributed ECC-DSS authentication scheme based on CRT-VSS and trusted computing}
\label{ECC-DSS authentication}
Among the schemes  about the distributed  authentication protocol, the robust Digital Signature Standard(DSS) signature  scheme proposed in \cite{gennaro1996robust} is significant, it can achieve security in different conditions. While the scheme is built on the basis of Shamir's scheme and the complexity of exponent arithmetic is high. Besides, the biggest problem is that the choice of the authentication  nodes is random. The lack of the validation of the secret shares in the node and the trusted property may lead to DoS and fault attack security problem.
The distributed ECC-DSS authentication scheme based on CRT-VSS and trusted computing proposed later will solve those problems together.
The  whole protocol consists of three phases, including trusted choice phase, signature generation phase and signature verification phase.

\subsection{trusted choice phase}
Suppose the request node $A$ apply for a signature for distributed authentication with message $M$. Here we don't focus on a password-protected mechanism \cite{passwordCCS11}  with secret sharing in some e-commercial scenarios(e.g., \cite{pirCCS2011}) and don't contain it explicitly. The details are as follows.
\begin{enumerate}
  \item Compute the digest of the message $M$ with some Hash function, e.g. $m$=SHA-1($M$).
  \item $A$ choose $n$ trusted authentication nodes $P_i(1 \leq i \leq n)$ to ask for distributed  authentication. The process is similar to that in distribution phase in CRT-VSS scheme in section~\ref{CRT-VSS}.
     \begin{itemize}
       \item A $\stackrel{B}\rightarrow$ $\{P_i\}$($P_i \in \mathbb{P}$): $id_D$,REQUEST.
       \item $P_i$ $\rightarrow$ A: $id_i$,$id_A$,$auth(id_i,id_A)$.\par
     \end{itemize}

       After the authentication process, request node $A$ choose $t$ qualified nodes as a coalition $C$ to finish the distributed authentication later.
  \item Request node $A$ broadcast the digest $m$ to the nodes in coalition $C$ as
    $A \stackrel{B}\rightarrow$ $\{P_i\}$($P_i \in \mathbb{C}$): $id_A$,$m$.
\end{enumerate}
\subsection{DSS signature generation phase}
This phase will do the generation work of the DSS signature. The scheme will be founded on the basis of Elliptic curve and the DSS signature\cite{gennaro1996robust}. Here are some notations and preliminaries before the phase. Suppose point $G$ is one point on the Elliptic curve $E_p(a,b) \in Z_n$ with the order $q$, $f$ is a big prime. The secret private key $S$ has been divided into $S_i$'s and  shared by node $P_i$ respectively. On the Elliptic curve, the public key denoted as $Q=SG$. Some necessary constrains: $0 \leq S_i \leq q$, $M_C \leq M \leq n$. The notation $\lambda \stackrel{CRT}{\longrightarrow}\{\lambda_i\}_{i \in \{1,\dots,n\}}$ represents that number $\lambda$ can be divided into shares $\{\lambda_i\}_{i \in \{1,\dots,n\}}$ with our CRT scheme proposed in section \ref{simple authentication} and  notation $\lambda =CRT_I\{\lambda_i\}_{i \in \{1,\dots,t\}}$ represents the reconstruction of $\lambda$ with $t$ shares  $\{\lambda_i\}_{i \in \{1,\dots,t\}}$.
The details  of the protocol are as follows.
\begin{enumerate}
  \item  Setup work. Choose random values $k,a \in Z_f$, $k \stackrel{CRT}{\longrightarrow}\{k_i\}_{i \in \{1,\dots,n\}}$, $a \stackrel{CRT}{\longrightarrow}\{a_i\}_{i \in \{1,\dots,n\}}$.
  \item Participant node $P_i$ broadcast $v_i=k_ia_i(mod\ n)$ and $w_i=a_iG(mod\ n)$. After collecting the $v_i$'s and $w_i$'s, each Participant node $P_i$ locally computes
      \begin{itemize}
        \item $ka=CRT_I\{v_i\}_{i \in \{1,\dots,t\}}(mod \ f)$
        \item $aG=CRT_I\{w_i\}_{i \in \{1,\dots,t\}}(mod\ f)$
        \item Compute $r=k^{-1}G=(ka)^{-1}(aG)(mod\ f)$
      \end{itemize}
  \item Each Participant $P_i$ broadcasts $sig_i=k_i(m+rS_i)$. After collecting, each Participant node $P_i$  computes global signature $sig=CRT_I\{sig_i\}_{i \in \{1,\dots,t\}}(mod\ f)$
  \item Output the $(r,s)$ as the signature for $m$.
\end{enumerate}
\subsection{signature verification phase}
\begin{enumerate}
  \item  Check the public parameters and  limitations, do the digest as m=SHA-1(M);
  \item  Do some preliminaries. $w=s^{-1}$,$u_1=mw(mod\ n)$, $u_2=rw(mod\ f)$;
  \item  Use the public key $Q=SG(mod \ f)$ to do the verification as $X=u_1G+u_2Q$. If $X=O$, then declare it as a fault signature ,or else suppose $X=(x,y)$, compute $v=x(mod\ f)$, if $v=r$ holds, then accept the signature, or else refuse.
\end{enumerate}

Here we will do some simple analysis about the correctness of the protocol and its security analysis will be done later in section~\ref{security analysis} separately. The proof of the correctness will have two aspect on the $s$ and $r$ respectively.
%If the signature $(r,s)$ on the message m is generated by the coalition of the t participant  nodes, then
   \begin{enumerate}
     \item Correctness of $r$. According to  the protocol above,
      \begin{spacing}{0}
   \begin{small}
   \begin{align*}
    \begin{split}
  r&=(ka)^{-1}(aG)(mod\ f)\\
  &=(CRT_I\{v_i\})_{i \in \{1,\dots,t\}}^{-1}\cdot CRT_I\{w_i\}_{i \in \{1,\dots,t\}}(mod\ f)\\
   &=(CRT_I\{k_ia_i\})_{i \in \{1,\dots,t\}}^{-1}\cdot CRT_I\{a_iG\}_{i \in \{1,\dots,t\}}\\
   &=CRT_I(\{k_ia_i\}_{i \in \{1,\dots,t\}})^{-1} \cdot CRT_I\{a_iG\}_{i \in \{1,\dots,t\}}(mod\ f) \\
   &\equiv (ka)^{-1}\cdot (aG)(mod\ f) \equiv k^{-1}G(mod\ f).
    \end{split}&
    \end{align*}
    \end{small}
    \end{spacing}
     \item  Correctness of  $sig$, combining the equation $r=k^{-1}G$ which has been proved in (1) above,
          \begin{spacing}{0.6}
   \begin{small}
   \begin{align*}
    \begin{split}
   sig&=CRT_I\{sig_i\}_{i \in \{1,\dots,t\}}(mod\ f)\\
      &=CRT_I\{k_i(m+rS_i)\}_{i \in \{1,\dots,t\}}(mod\ f)\\
      &=CRT_I\{k_im\}_{i \in \{1,\dots,t\}}+CRT_I\{k_iS_ir\}_{i \in \{1,\dots,t\}}(mod\ f) \\
      &\equiv km+rCRT_I\{k_i\}_{i \in \{1,\dots,t\}}CRT_I\{S_i\}_{i \in \{1,\dots,t\}}\\
      &=km+rkS(mod\ f)\\
      &=k(m+Sr)(mod\ f).
    \end{split}&
    \end{align*}
    \end{small}
    \end{spacing}

     It is the right shape of the DSS signature  $s$.

   \end{enumerate}

\section{some discussion about shares refreshing scheme}
\label{refreshing scheme}
\subsection{refreshing discussion on Shamir's scheme}
\label{refreshing-shamir}

In the scheme proposed in \cite{LidongZhou} and other literature, the traditional refreshing scheme of the shares use the $(\oplus,\oplus)$-homomorphism property for Shamir's scheme\cite{Shamir}.
Each  participant nodes  randomly generate refreshing shares and divide them into subshares by randomly chosen  $t$-degree random Shamir's polynomial with sharing of 0. After participants finish exchanging the subshares, they refresh the shares with subshares by addition operation. 
 %\begin{enumerate}
%   \item Each  participant node $P_i(1 \leq i \leq n)$ randomly generate refreshing shares $w_i(1 \leq i \leq n)$.   Let $f_i$ be the chosen  $t$-degree random Shamir's polynomial with sharing of 0, hence $f_i(0)=w_i$;
%   \item Furthermore, suppose $w_i$ generate the refreshing subshares  according to the polynomial $f_i$, notated $w_i \stackrel{f_i}{\longrightarrow}\{w_{ij}\}_{j \in \{1,\dots,n\}} $, the  newly generated $w_{ij}$'s  are called refreshing subshares in \cite{LidongZhou}.
%   \item  Then all subshares $w_{ij}$'s is distributed to participant $P_j$ through a secure link.
%   \item   When participant $P_j$ gets the subshares $w_{ij}(i \neq j)$ from all other nodes , it can compute a new share from these refreshing subshares and its old share $S_j^{'} = S_j + \sum_{1 \leq i \leq n}{w_{ij}}$.
% \end{enumerate}
The literature before all designs the share refreshing scheme with the Shamir's threshold scheme in \cite{LidongZhou},\cite{xiongyan}. They all set the refreshing shares 0, then generate the refreshing subshares.  Here we will extend such refreshing schemes based on  $(\oplus,\oplus)$-homomorphism property for Shamir's scheme.

\begin{theorem}
The refreshing scheme based on  $(\oplus,\oplus)$-homomorphism property for Shamir's scheme\cite{Shamir} is correct as follows. For each participant $P_i$, Refreshing shares(RS) constraints and subshares generation:
\label{(plusShamir)}

 \begin{equation*}
\left\{
\begin{array}{l@{\;\;}l}
\sum_{i=1}^{n}w_{i}=0   \quad (a)\\
w_i \stackrel{f_i}{\longrightarrow}\{w_{ij}\}_{j \in \{1,\dots,n\}}
%w_i=F_I{ (w_{i1},w_{i2}, \dots, w_{in})}(subshares)
\end{array}
\right.
\end{equation*}
Refreshing operation:

  $$S_j^{'} = S_j + \sum_{1 \leq i \leq n}{w_{ij}}\quad (b)$$

\end{theorem}
\begin{proof}The proof is easy and similar to \cite{Shamir}, thus we neglect it due to the space limitation here.
\end{proof}

Obviously, the correctness of refreshing scheme above is built on the assumption that all the participants are correct and honest. As mentioned before that the lack of the trusted assurance of the generation of the refreshing subshares will result to the retortion problems of the secret $S$.
%In the scheme proposed below, we will take the trusted computing technique to eliminate the problem and design the refreshing scheme with the extended refreshing  scheme on Shamir's scheme illustrated above.
%
%
%\begin{enumerate}
%\item \textbf{Trusted choice.} The trusted choice of refreshing shares provider set as before, notated as $A=\{P_{i_j}\}_{ j \in \{1,\dots,k\}}$, $| A^{\oplus} | =k$, and $1 \leq k \leq n$.
%\item \textbf{refreshing shares generation.} All the elected trusted participant $\{P_{i_j}\}_{ j \in \{1,\dots,k\}}$ generate the refreshing shares $\{w_{i_{j}}\}_{ j \in \{1,\dots,k\}}$ according to condition limitation as $\sum_{i=1}^{k}w_{i_j}= 0(mod \ M)$;
%\item \textbf{refreshing subshares generation.}  Furthermore,suppose $w_{i_j}$ generate the refreshing subshares  according to the polynomial $f_{i_j}$, notated $w_{i_j} \stackrel{f_{i_j}}{\longrightarrow}\{w_{iq}\}_{q \in \{1,\dots,n\}} $. Then  deliver $w_{iq}$'s to the other node respectively through secure channel.
%\item \textbf{refreshing operation.}  When participant $P_j$ gets the subshares $w_{ij}(i \neq j)$ from all other nodes , it can compute a new share from these refreshing subshares and its old share, $S_j^{'} = S_j +\sum_{1 \leq i \leq k}{w_{ij}}(mod\ m_j)$.
%
%\end{enumerate}

\subsection{refreshing discussion on CRT-VSS scheme}
\label{refreshing-CRT}
Now we discuss the refreshing scheme on CRT-VSS scheme, combining the $(\oplus,\oplus)$ and $(\otimes,\otimes)$ property in Lemma~\ref{CRTproperty}, we can easily get two trivial  distributed refreshing schemes in  the similar way as Shamir's scheme in Theorem~\ref{(plusShamir)}, which can furthermore deduce a general mixed scheme. Here We skip the illustration of such two trivial schemes and deliver the mixed scheme below.
%\begin{proposition}
%  Similarly, the scheme based on $(\oplus,\oplus)$-homomorphism property for  CRT scheme is correct as follows.
%\label{plusCRT}
%RS constraints and subshares generation:
% \begin{equation*}
%\left\{
%\begin{array}{l@{\;\;}l}
%\sum_{i=1}^{n}w_{i}\equiv 0&(mod \ M)\quad (a)\\
%\{w_{i}\}_{i \in \{1,\dots,n\}} \not \equiv 0&(mod\ M)\quad (b)\\
%w_i \stackrel{CRT}{\longrightarrow}\{w_{ij}\}_{j \in \{1,\dots,n\}}
%%w_i=F_I{ (w_{i1},w_{i2}, \dots, w_{in})}(subshares)
%
%\end{array}
%\right.
%\end{equation*}
%here  $M=\prod_{i=1}^{n}m_i$.Refreshing operation:
%
%  $$S_j^{'} = S_j + \sum_{1 \leq i \leq n}{w_{ij}}\quad (c)$$
%
%\end{proposition}
%
%
%
%\begin{proposition}The scheme based on $(\otimes,\otimes)$-homomorphism property  for  CRT scheme is correct as follows.
%\label{(timeCRT)}
%RS constraints and subshares generation:
% \begin{equation*}
%\left\{
%\begin{array}{l@{\;\;}l}
%\prod_{i=1}^{n}w_{i}\equiv 1&(mod \ M)\quad (a)\\
%\{w_{i}\}_{i \in \{1,\dots,n\}} \not \equiv 1&(mod\ M)\quad (b)\\
%w_i \stackrel{CRT}{\longrightarrow}\{w_{ij}\}_{j \in \{1,\dots,n\}}
%%w_i=F_I{ (w_{i1},w_{i2}, \dots, w_{in})}(subshares)
%\end{array}
%\right.
%\end{equation*}
%here  $M=\prod_{i=1}^{n}m_i$.Refreshing operation:
%
%  $$S_j^{'} = S_j \times \prod_{1 \leq i \leq n}{w_{ij}}\quad (c)$$
%
%\end{proposition}
 % In the  scheme below, $(n-k)(1\leq k \leq n-1)$ nodes generate  refreshing shares following $(\oplus,\oplus)$ scheme, while $k$ nodes left with  $(\otimes,\otimes)$ scheme.
\begin{proposition}
\label{MIXCRT}
The $(\oplus,\oplus)$ and $(\otimes,\otimes)$-mixed  scheme based on CRT-VSS scheme is correct as follows.
RS constraints and subshares generation:
 \begin{equation*}
\left\{
\begin{array}{l@{\;\;}l}
\prod_{i=1}^{k}w_{i}\equiv 1,\sum_{i=k+1}^{n}w_{i}\equiv 0&(mod \ M)\\
\{w_{i}\}_{i \in \{1,\dots,k\}} \not \equiv 1&(mod\ M)\quad (a)\\
\{w_{i}\}_{i \in \{k+1,\dots,n\}} \neq 0&(mod\ M)\quad (b)\\
w_i \stackrel{CRT}{\longrightarrow}\{w_{ij}\}_{j \in \{1,\dots,n\}}
%w_i=F_I{ (w_{i1},w_{i2}, \dots, w_{in})}(subshares)
\end{array}
\right.
\end{equation*}
Refreshing operation:
\begin{small}
    $$S_j^{'} = S_j \times \prod_{1 \leq i \leq k}{w_{ij}}+\sum_{i=k+1}^{n}w_{ij}\quad (c)$$

\end{small}
\end{proposition}

%Though The trivial schemes above can achieve reconstruction without retortion, but the refreshing has little function in the arithmatic of modulus $m_i$.
%Without loss of generality, we will do the simple proof about the Theorem~\ref{MIXCRT},
%\begin{proof}
%for all nodes $P_i(1 \leq i \leq n)$, the reconstruction equation:
%\begin{center}
%   \begin{spacing}{0}
%   \begin{small}
%   \begin{flalign*}
%    \begin{split}
% \hspace{10mm} &CRT_I(\{S_j^{'}\})_{j \in \{1,\dots,n\}}\\
% &= CRT_I(\{\prod_{1 \leq i \leq n}S_j\cdot w_{ij}+\sum_{k+1 \leq i \leq n}{w_{ij}}\})_{j \in \{1,\dots,n\}}\\
% &=CRT_I(\{S_i\}_{i \in \{1,\dots,n\}})\cdot\prod_{1 \leq i \leq k}CRT_I(w_{ij})_{j \in \{1,\dots,n\}}\\
% &\quad +\sum_{1 \leq i \leq n}CRT_I(\{w_{ij}\})_{j \in \{1,\dots,n\}}(mod\ M)\\
% &=S\cdot \prod_{1 \leq i \leq k}w_i+\sum_{k+1 \leq i \leq n}w_i(mod\ M)\\
% &\equiv S\times1+0(mod\ M)\equiv S(mod\ M)
%    \end{split}&
%    \end{flalign*}
%    \end{small}
%    \end{spacing}
% \end{center}
We conduct modulus arithmatic operation,
\begin{center}
   \begin{spacing}{0}
   \begin{small}
   \begin{flalign*}
    \begin{split}
 \hspace{10mm} &S_j^{'}\ mod\ m_i \\
 &= S_j \times \prod_{1 \leq i \leq k}{w_{ij}}+\sum_{k+1 \leq i \leq n}w_{ij}(mod\ m_i)\\
 &\equiv S_j \times (\prod_{1 \leq i \leq k}{w_{ij}}\ mod\ m_i)+(\sum_{k+1 \leq i \leq n}w_{ij}\ mod\ m_i)\\
 &\equiv S_j \times 1+0 \equiv S_j
    \end{split}&
    \end{flalign*}
    \end{small}
    \end{spacing}
 \end{center}

 Obviously, the refreshing results has little function under the modulus $m_i$ operation. So the three trivial refreshing schemes aren't correct refreshing schemes for CRT-VSS scheme.

Unluckily, besides the failure with  trivial schemes above, we can have the furthermore upsetting conclusion as follows.
\begin{theorem}
\label{distribut-failure}
  Without changing the prime parameters set $\{m_i\}_{i \in \{1,\dots,n\}}$, the refreshing scheme without the retortion problems can not avoid the refreshing failure under the modulus $m_i$ operation, formally denoted as

   \begin{equation*}
\left\{
\begin{array}{l@{\;\;}l}
S =CRT_I\{S_i\}_{i \in \{1,\dots,n\}}\quad (a)\\
S=CRT_I\{S_{i}^{'}\}_{i \in \{1,\dots,n\}}\quad (b)\\
%w_i=F_I{ (w_{i1},w_{i2}, \dots, w_{in})}(subshares)
\end{array}
\right.
\end{equation*}
$\Rightarrow S_i \equiv S_{i}^{'}\ (mod\ m_i)\quad (1 \leq i \leq n).$

\end{theorem}

\begin{proof}
  The proof can be finished in two different ways. Here we just provide one due to the constraints of limitation.
  % The denotations are all like before.
%
%
%  \textbf{The first way.} According to the equation (a),(b), we can get\\
%\begin{equation*}
%\left\{
%\begin{array}{c@{\;\;}l}
%  S \equiv \sum_{k=1}^{t}S_{i}M_{C,i}^{'} M_{C \backslash \{i\}}\quad (mod \ M_C)\\
%  S\equiv \sum_{k=1}^{t}S_{i}^{'}M_{C,i}^{'} M_{C \backslash \{i\}}\quad (mod \ M_C)
%\end{array}
%\right.
%\end{equation*}
% the minus of the two equation can be:
% \begin{center}
%    $ \sum_{k=1}^{t}(S_{i}^{'}-S_i)M_{C,i}^{'} M_{C \backslash \{i\}} \equiv 0\quad (mod \ M_C)$
% \end{center}
% Let $\triangle S_i$ denotes the term $S_{i}^{'}-S_i$, the problem can be turned into diophantine equation:
%  \begin{center}
%    $ \sum_{k=1}^{t}\triangle S_i M_{C,i}^{'} M_{C \backslash \{i\}} \equiv 0\quad (mod \ M_C)$
%  \end{center}
%  This diophantine equation of $\triangle S_i$ above only has solutions following $\triangle S_i \equiv 0(mod\ M_C)$. So the conclusion can be
%  proved.
 According to the equation (a),(b),
  suppose the reconstruction result before the modulus operation
 \begin{equation*}
\left\{
\begin{array}{c@{\;\;}l}
  A=\sum_{k=1}^{t}S_{i}M_{C,i}^{'} M_{C \backslash \{i\}}\\
  A^{'}=\sum_{k=1}^{t}S_{i}^{'}M_{C,i}^{'} M_{C \backslash \{i\}}
\end{array}
\right.
\end{equation*}
and
 \begin{equation*}
\left\{
\begin{array}{c@{\;\;}l}
  S \equiv S_i(mod\ m_i)\\
  S^{'}\equiv S_{i}^{'}(mod\ m_i)
\end{array}
\right.
\end{equation*}
If the reconstruction without retortion holds, then\\
$ A \equiv A^{'}(mod\ M_C)\Rightarrow \exists k \in \mathbb{Z},A=kM_C+A^{'}$\\
hence, we can get
\begin{center}

   \begin{spacing}{0}
   \begin{small}
   \begin{align*}
    \begin{split}
  A(mod\ m_i)&=kM_C+A^{'}(mod\ m_i)\\
  &\equiv (kM_C\ mod \ m_i)+A^{'}\ mod\ m_i\\
  &\equiv 0+S_{i}^{'}(mod\ m_i)\equiv S_{i}^{'}.
    \end{split}&
    \end{align*}
    \end{small}
    \end{spacing}
\end{center}

So such refreshing scheme results to little function by modulus operation with $\{m_i\}_{i \in \{1,\dots,n\}}$.
\end{proof}

According to the Theorem~\ref{distribut-failure}, a correct refreshing scheme for CRT-VSS scheme need to change the prime parameters set $\{m_i\}_{i \in \{1,\dots,n\}}$, the essence of the necessity is the newly computation of a new share $S_{i}^{'}$ modulus new $m_i$'s.  And this will bring to some requirements as follows.
\begin{itemize}
  \item The right computation of new share $S_{i}^{'}$ modulus new $m_i$'s; The computation should be right and dependable, this can be done by a Trusted Third Party(TTP).
  \item The right and secure delivery of the new shares; this can be easily done by secure channel or by cryptography mechanisms.
  \item The right update of the shares and remove of the old ones. This can be assured by the normal process of the protocol, so it can be ensured by the trusted computing techniques.
\end{itemize}

%So we can get the framework of the secure refreshing scheme for CRT scheme in Fig~\ref{fig:noattackselect}.
%\begin{figure}[h]
% \centering
% \fbox{
%\begin{minipage}{3.2 in}
%\begin{enumerate}
%\item \textbf{New shares computation}. New $m_i$'s computed by trusted party, denoted as $T$, which can be trusted secret distributor or some TTP. choose a different primes set $\{m_{i}^{'}\}_{i \in \{1,\dots,n\}}$, compute $S_{i}^{'}= S\ mod \ m_{i}^{'}$,
%\item \textbf{Trusted choice.} The trusted choice of trusted participant set, notated as $A=\{P_{\lambda_i}\}_{ i \in \{1,\dots,k\}}$, $| A^{\oplus} | =k$, and $1 \leq k \leq n$.
%\item  \textbf{Deliver the new shares}. $m_{i}^{'}$ in (1) to participant $P_{\lambda_i}$ securely;
%\item \textbf{Refreshing operation.}  When participant $P_{\lambda_i}$ gets the new shares $S_{i}^{'}$, it remove the old share $S_i$, then replace it with $S_{i}^{'}$.
%
%\end{enumerate}
%\end{minipage}
%}
% \caption{framework of refreshing scheme for CRT-VSS scheme}
% \label{fig:noattackselect}
%\end{figure}

\subsection{comparisosn between Shamir's  and CRT scheme }
From the analysis in Section~\ref{refreshing-shamir} and \ref{refreshing-CRT}, we note  that there exist some differences between CRT and Shamir's schemes. 
Thus we  summarize these comparison results in Table~\ref{comparison}.
\begin{table}
\centering
\caption{comparison between Shamir's scheme and CRT scheme }
\begin{tabular}{|l|c|c|}
\hline
 Comparison Items & Shamir's & CRT  \\
\hline
Computation complexity & $O(tlog^2t)$ & $O(t)$\\
\hline
Reconstruction parameter related & No & Yes. \\
\hline
Fit distributed refreshing & Yes & No.\\% distributed in nodes \\
\hline
Fit TTP refreshing  & No & Yes \\
\hline
Homomorphism property& $(\oplus,\oplus)$ &$(\oplus,\oplus)$\\
& & $(\otimes,\otimes)$ \\
\hline
Secure shares-product  concision  & Not good &Good  \\
\hline
\end{tabular}
\label{comparison}
\end{table}

\section{security  analysis  }
\label{security analysis}
In this section, we will prove the security of simple authentication framework between two nodes with CSP \cite{Hoare85} model and rank function \cite{Schneider98} . Then extend the secure assurance to  multi-run and arbitrary size authentication case. After that we model our CRT-VSS scheme, ECC-DSS authentication scheme based on trusted computing and prove the security of the schemes.

The TPM can provide the authentication property towards the remote party in a secure and private way. The security of our protocol is built on the  assumptions including: 1)TPM can not be compromised by the intruder maintaining a right property. By now, there is no successful way to compromise the TPM without any damage of the right property of TPM.  Whatever happens, the TPM can record and reveal the  property of  software/hardware within the  node. 2) The security of DAA and PBA protocol. It is mentioned in \cite{DAA} and \cite{PBA} that the two protocol was secure under different assumptions. Besides, the cryptographic primitives, including encryption $\{\cdot\}_k$ and signature $S_{sk}(\cdot)$ we employed are secure.
 We will formally model  the protocol and  security property based on communicating sequential processes(CSP)\cite{Hoare85}, and use the rank function\cite{Schneider98} to do some analysis and proof furthermore.
\subsection{The preliminaries}
\subsubsection{The network}
The network considered below will consist of  several different parties: Distributor denoted as $D$, Applicant for distributed authentication signature denoted as $A$, Users acted as the participants with the secret shares denoted as set $USER$ and the intruder denoted as $ENEMY$.
The model of the network is standard, based on the Dolev-Yao model. The $ENEMY$ is in full control of the network, and all communication passes through it. The attacker may choose to block or redirect messages, send messages of its choice, and imitate as other users.
The capabilities of the attacker are bounded by a finite set of deductions, by which it can generate new messages from the set of messages it already knows.
 Let ``generates'' relation  $X \vdash m$ denotes that new message $m$ is deduced from a set $X$ of messages.  Here if m,n are messages, $k$ is a key and $k^{-1}$ is the inverse of $k$, then the smallest $\vdash$ relation satisfies.  $\{m,n\} \vdash m.n/m/n$,
   $\{m,k\} \vdash \{m\}_k/S_k(m)$, $\{\{m\}_k,k^{-1}\} \vdash m$, etc.
%In our context, the attacker is capable of the following three deductions related to the TPM commands used in the protocol, in addition
\subsubsection{Authentication}
 Authentication can be captured in terms of events whose occurrence guarantees the prior occurrence of other messages. The key property is concerned with precedence between events. Conditions  can  be expressed as trace specifications on $NET$, requiring that no event from a set $T$ has occurred unless another event from a set $R$ has previously occurred, which can be illustrated as:
 $NET\ sat \ R \ precedes\ T$.
The rank function approach proposed in \cite{Schneider98} assigns a rank 0 to the messages that should remain secret during the protocol and a rank 1 to the messages that should be public and might be get held by the enemy.
The central rank theorem from \cite{Schneider98}  gives four conditions the rank
function must satisfy in order to maintain the rank with the rank function.

\subsection{Security analysis for simple authentication framework between two nodes}
\subsubsection{Simplification for protocol}
The protocol framework illustrated in Figure~\ref{fig:frameworkAUTH} can be simplified as Figure~\ref{fig:simplification}.
\begin{figure}[h]
 \centering
 \fbox{
\begin{minipage}{3.2 in}
1. $P_i$ $\rightarrow$ $P_j$: $id_i$,$id_j$,INIT\\
2. $P_j$ $\rightarrow$ $P_i$: $id_j$,$id_i$,$n_1$\\
3. $P_i$ $\rightarrow$ $P_j$: $id_i$,$id_j$,$im$,$DS_i$,$PS_i$,K\_INFO\_I\par
4. $P_j$ $\rightarrow$ $P_i$: $id_j$, $id_i$,$S_{pk_{AIK_J}}$(K\_INFO\_J$||id_j$),$\{id_j||n_{2}\}_{k_{ij}}$ \par
5. $P_i$ $\rightarrow$ $P_j$: $id_i$,$id_j$,$\{id_i||n_2\}_{k_{ij}}$ \par
\end{minipage}
}
 \caption{$P_i$ request authentication to node $P_j$ }
 \label{fig:simplification}
\end{figure}

\subsubsection{Protocol modeling}
single run and node I as the initiator and node J as the responder.

Node I can be modeled as follows.
 \begin{spacing}{0.6}
   \begin{small}
   \begin{align*}
    \begin{split}
  USER(I)&= \Box_{j \in P}trans.i.j.id_i.id_j.INIT\\
         &\to recv.i.j.id_j.id_i.n_1\\
         &\to trans.i.j.id_i,id_j.im,DS_i,PS_i, K\_INFO\_I\\
         &\to recv.i.j.id_j, id_i,S_{pk_{AIK_J}}(K\_INFO\_J||id_j),\{id_j||n_{2}\}_{k_{ij}}\\
         &\to trans.i.j.id_i,id_j,\{id_i||n_2\}_{k_{ij}}\to STOP.
    \end{split}&
    \end{align*}
    \end{small}
    \end{spacing}

Node J can be modeled in a similar way.
as follows.
 \begin{spacing}{0.6}
   \begin{small}
   \begin{align*}
    \begin{split}
  USER(J)&=recv.j.i.id_i,id_j,INIT\\
         &\to trans.j.i.id_j,id_i,n_1\\
         &\to recv.j.i.id_i,id_j,im,DS_i,PS_i, K\_INFO\_I\\
         &\to validate\&Signal.Trust.i.n_1.n_2.\\
         &\to trans.j.i.id_j,id_i,S_{pk_{AIK_J}}(K\_INFO\_J||id_j),\{id_j||n_{2}\}_{k_{ij}}\\
         &\to recv.j.i.id_i,id_j,\{id_i||n_2\}_{k_{ij}}\to STOP.
    \end{split}&
    \end{align*}
    \end{small}
    \end{spacing}
 The following process $ENEMY$ characterize the behavior of the attacker in possession of a set $X$ of messages.
 \begin{spacing}{0.6}
   \begin{small}
   \begin{align*}
    \begin{split}
  ENEMY(X)&=trans.i.j.m \rightarrow ENEMY(X \cup \{m\})\\
          & \Box\quad  recv.i.j.m \rightarrow ENEMY(X).
    \end{split}&
    \end{align*}
    \end{small}
    \end{spacing}

The entire network is thus modeled by process $NET$, a parallel composition of $USER(I)$, $USER(J)$ and $ENEMY$ . The process is synchronized on trans and rec:
 \begin{spacing}{0.6}
   \begin{small}
   \begin{align*}
    \begin{split}
 NET=(USER(I) ||| USER(J)) \underset{\{trans,recv\}}{\parallel}ENEMY
    \end{split}&
    \end{align*}
    \end{small}
    \end{spacing}
The objective of the protocol can be expressed as $R=\{trans.i.j.id_i,id_j,im,DS_i,PS_i,$K\_INFO\_I$\}$ precedes $T=\{signal.Trust.i.n\}$.

The rank function for the authentication framework between the two nodes are as in Figure~\ref{fig:rankfunction}.
\begin{figure}[htb]
 \centering
 \fbox{
\begin{minipage}{3 in}
User and host: $\rho(id_i)=1$,$\rho(INIT)=1$\par
Nonces:$\rho(n_1)=1$,$\rho(n_2)=0$ \par
Integrity metrics: $\rho(im)=1$ \par
TPM signatures: $\rho(DS_i)=0$,$\rho(PS_i)=0$ \par
K\_INFO: $\rho(g)=1$,$\rho(g^x)=0$,$\rho(g^y)=0$ \par
Encryption: $\rho(\{m\}_{k_{ij}}=\rho(k_{ij}))$ \par
AIK Signature: $\rho(S_{pk_{AIK_J}}$(K\_INFO\_J$||id_j))=0$ \par
Concatenations:$\rho(m_1,m_2)=min\{\rho(m_1),\rho(m_2)\}$ \par
Events: $\rho(trans.i.j.m)=\rho(rec.i.j.m)=\rho(m)$\par
Signals: $\rho(Signal.Trust.i.n_1.n_2)=0$
\end{minipage}
}
 \caption{Rank function for the authentication between node $P_i$ and $P_j$ }
 \label{fig:rankfunction}
\end{figure}

Checking the four conditions above is easy and just routine. We neglect the details due to the limitation of paper.
 With the rank function in the  figure above, we can see that:
\begin{itemize}
  \item  The messages in the initial set $X$ known to the $ENEMY$ process are of rank 1.
 \item Any message generated by $ENEMY$ from $X$ under the deductions rules has a rank pub whenever messages in $X$ have a rank 1.
 \item  $\rho(Signal.Trust.i.n_1.n_2)$ in $R$ has a rank 0.
 \item  When $R=\{trans.i.j.id_i,id_j,im,DS_i,PS_i,$K\_INFO\_I$\}$ is blocked, user can never give out a message of rank 0 unless it has previously received one.
\end{itemize}
So according to the central theorem of rank function, the rank function above can maintain the rank well. Furthermore, the authentication property of the simple authentication framework can be proved secure well.

\subsection{Security analysis for multi-run and arbitrary size authentication framework}
\label{abritrary}
In this subsection we  show  that  if there is an attack up on the authentication framework of arbitrary size with multi-run  given in the subsection before, then there is an attack  up on  the  small  framework described above  with a single initiator A and  a single responder B.
Similar work with the classical Needham-Schroeder Public-Key Protocol can be seen in  \cite{breaking}. The process of proof can be carried out by two aspects from the attacks upon responder and initiator respectively.

\subsubsection{Attacks upon the responder}
Consider a run $\alpha$  where the intruder imitates the initiator $A$ to attack the responder $B$ as in Figure~\ref{fig:alpha}.
\begin{figure}[h]
 \centering
 \fbox{
\begin{minipage}{3.2 in}
Message $\alpha$.1. $I(A)$ $\rightarrow$ $B$: $A$,$B$,INIT\\
Message $\alpha$.2. $B$ $\rightarrow$  $I(A)$: $B$,$A$,$n_1$\\
Message $\alpha$.3.  $I(A)$ $\rightarrow$ $B$: $A$,$B$,$im_A$,$DS_A$,$PS_A$,$g$,$g^x$ \par
Message $\alpha$.4. $B$ $\rightarrow$  $I(A)$: $B$,$A$,$g^y$,$\{B||n_{2}\}_{k_{AB}}$ \par
Message $\alpha$.5.  $I(A)$ $\rightarrow$ $B$: $A$,$B$,$\{A||n_2\}_{k_{AB}}$ \par
\end{minipage}
}
 \caption{Intruder imitates initiator A to attack responder B }
 \label{fig:alpha}
\end{figure}

In order to attack upon the responder $B$ successfully, the intruder needs to produce message 1,3 and 5. Message 1 can be produced easily. As to the message 3, the intruder can produce it by replay the signature $DS_A$ and $PS_A$ by eavesdropping the message with the signatures inside. While the get of the message 5 can be rather difficult since the session key $k_{AB}$ only established between the real A and B party. So the intruder cannot decrypt the message $\alpha.4$, so he learns neither $n_2$ nor $\{A||n_2\}_{k_{AB}}$ from message $\alpha.4$, denoted as:$\neg$(Intruder learns $n_2$\ $\lor$\ $\{A||n_2\}_{k_{AB}}$ from $\alpha.4$).

According to the analysis delivered in \cite{breaking}, the way to solve this problem for intruder is to replay the encrypted part of message  $\alpha.4$ in message 4 in another run called $\beta$, and learns the  $n_2$ or  $\{A||n_2\}_{k_{AB}}$ from message $\beta.5$. Note that the initiator of run $\beta$ must be $A$. Hence the form of the run $\beta$ should be as follows in Figure~\ref{fig:beta}:
\begin{figure}[h]
 \centering
 \fbox{
\begin{minipage}{3.2 in}
Message $\beta$.1. $A$ $\rightarrow$ $I(B)$: $A$,$B$,INIT\\
Message $\beta$.2. $I(B)$ $\rightarrow$  $A$: $B$,$A$,$n_1$\\
Message $\beta$.3.  $A$ $\rightarrow$ $I(B)$: $A$,$B$,$im_A$,$DS_A$,$PS_A$,$g$,$g^x$ \par
Message $\beta$.4. $I(B)$ $\rightarrow$  $A$: $B$,$A$,$g^y$,$\{B||n_{2}\}_{k_{AB}}$ \par
Message $\beta$.5.  $A$ $\rightarrow$ $I(B)$: $A$,$B$,$\{A||n_2\}_{k_{AB}}$ \par
\end{minipage}
}
 \caption{The form of auxiliary run $\beta$ }
 \label{fig:beta}
\end{figure}
Now we see that the  intruder learns  the component  $im_A$,$DS_A$,$PS_A$,$g$,$g^x$ from
Message $\beta$.3. and replays this in  $\alpha.3$, learns  the component  $\{A||n_2\}_{k_{AB}}$ from
Message $\beta$.5. and replays this in  $\alpha.5$. So only these two runs  are necessary for the intruder to learn all the knowledge  it uses in the attack upon the responder.  Thus if the intruder can imitate the initiator A to attack  the responder B  then such  an attack  would  have  been found  by  considering  the small  system above.

\subsubsection{Attacks upon the initiator}
Similar analysis can be done and we neglect it here.

Combining these two aspects, we can get the conclusion that our scheme with multi-run and arbitrary size can also be secure under authentication property.

\subsection{Security analysis  for the complete protocols in this paper}
\subsubsection{Distributing phase for CRT-VSS scheme}
The model of process $T$ characterize trusted election before the distributing process, $D$ dealing with the distributing procedure. Process $USER$ represents the behavior of node  during the first phase in the CRT-VSS scheme.
 \begin{spacing}{0.6}
   \begin{small}
   \begin{align*}
    \begin{split}
T_I(i)&=trans.D.i.id_D,AUTH,S_{sk_{AIK_D}(id_D)}\\
    &\to recv.D.i.id_i,auth(id_i,id_D)\\
    &\to validate\& signal.Trust..\\
    &\to STOP\\
T_R(i)&=recv.D.i.id_i,id_D,AUTH\\
      &\to trans.D.i.id_D,id_i,auth(id_D,id_i)\\
      &\to recv.D.i.id_i,id_D,AGREE\\
      &\to STOP.
    \end{split}&
    \end{align*}
    \end{small}
    \end{spacing}
\begin{spacing}{0.6}
   \begin{small}
   \begin{align*}
    \begin{split}
D(i)&= \Box_{i \in C}trans.D.i.id_D,id_i,S_{sk_{AIK_D}(id_D)},enc_{k_{iD}}(S_i)\\
      &\to STOP\\
      & \Box \Box_{i \in P}trans.D.i.id_D,id_i,S_{sk_{AIK_D}(id_D||p_i||g_i||m_i)},PUBLIC\\
      &\to STOP.
    \end{split}&
    \end{align*}
    \end{small}
    \end{spacing}
 \begin{spacing}{0.6}
   \begin{small}
   \begin{align*}
    \begin{split}
USER_R(i)&=recv.i.D.id_D,id_i,AUTH,S_{sk_{AIK_D}(id_D)}\\
    &\to validate(S_{sk_{AIK_D}(id_D)})\\
    &\to trans.i.D.id_i,id_D,auth(id_i,id_D)\\
    &\to STOP\\
       \end{split}&
    \end{align*}
    \end{small}
    \end{spacing}
 \begin{spacing}{0.6}
   \begin{small}
   \begin{align*}
    \begin{split}
USER_I(i)&=trans.i.D.id_i,id_D,AUTH\\
      &\to recv.i.D.id_D,id_i,auth(id_D,id_i)\\
      &\to validate(auth(id_D,id_i))\&signal.Trust...\\
      &\to trans.i.D.id_i,id_D,AGREE\\
      &\to STOP.
    \end{split}&
    \end{align*}
    \end{small}
    \end{spacing}

The process $T_R(i)$ and $USER_I(i)$ are similar to $T_I(i)$ and $USER_R(i)$.
 The entire network is thus modeled as follows.
 \begin{spacing}{0.6}
   \begin{small}
   \begin{align*}
    \begin{split}
      D&=|||_{i \in C}(D(i)  ),\\
     T(i)&= ((T_I(i) ||| T_R(i))), T=|||_{i \in \mathbb{P}}(T(i)  ),\\
    USER(i)&=((USER_I(i) ||| USER_R(i))),\\
    USER&=|||_{i \in C}USER(i),\\
    DIS&=(USER|||T) \rightarrow D,\\
 NET&=(DIS) \underset{\{trans,recv\}}{\parallel}ENEMY
    \end{split}&
    \end{align*}
    \end{small}
    \end{spacing}

The objective of the distribution phase scheme is $R=$\\$\{trans.i.D.id_i,auth(id_i,id_D)$$\land \ trans.D.i.id_D,auth(id_D,id_i)\}$ precedes $T=\{trans/recv.i.D.id_i,AGREE\}$.

\subsubsection{Reconstruction phase for CRT-VSS scheme}

In this phase, the participant exchange the secret shares with each other. Each node send and receive and shares at the same time as initiator and responder respectively.
 \begin{spacing}{0.6}
   \begin{small}
   \begin{align*}
    \begin{split}
USER_R^{I}(J)&=\Box_{J \in C}recv.i.j.id_j,id_i,AUTH\\
    &\to trans.i.j.id_i,id_j,auth(id_i,id_j)\\
    &\to recv.i.j.id_j,id_i,\{S_j\}_{k_{ij}}\\
    &\to STOP\\
 \end{split}&
    \end{align*}
    \end{small}
    \end{spacing}
\begin{spacing}{0.6}
   \begin{small}
   \begin{align*}
    \begin{split}
USER_I^{I}(J)&=\Box_{I \in C}trans.i.j.id_i,id_j,AUTH\\
      &\to recv.i.j.id_j,id_i,auth(id_j,id_i)\\
      &\to validate(auth(id_j,id_i))\&signal.Trust...\\
      &\to trans.i.j.id_i,id_j,\{S_i\}_{k_{ij}}\\
      &\to STOP.
    \end{split}&
    \end{align*}
    \end{small}
    \end{spacing}

The entire network  for the reconstruction scheme is thus modeled as follows.
 \begin{spacing}{0.6}
   \begin{small}
   \begin{align*}
    \begin{split}
    USER^{I}&=\Box_{J \in C}(|||USER_I^{I}(J)) ||| (|||USER_R^{I}(J)),\\
    USER&=|||_{I \in C}(USER^I),\\
 NET&=USER \underset{\{trans,recv\}}{\parallel}ENEMY
    \end{split}&
    \end{align*}
    \end{small}
    \end{spacing}
Similarly, the objective of the reconstruction scheme is as follows: $R=\{trans.j.i.auth(id_i,id_j)\}$ precedes $T=\{trans.i.j.\{S_i\}_{k_{ij}}\}$. Here we neglect the model details like above  due to the limitation of paper.

\subsubsection{Distributed ECC-DSS authentication scheme based on trusted computing}
In the scheme, the node A will request for the authentication and signature from the participant nodes. Each node will send and receive  shares at the same time as initiator and responder respectively. The model can be described as follows.

 \begin{spacing}{0.6}
   \begin{small}
   \begin{align*}
    \begin{split}
REQ_A(I)&=\Box_{i \in P}trans.A.i.id_A,id_i,REQU\\
       &\to recv.A.i.id_i,id_A,AGREE\\
       &\to trans.A.i.id_A,id_i,AUTH\\
       &\to recv.A.i.id_i,id_A,auth(id_i,id_A)\\
       &\to validate\&signal.Trust....\\
       &\to trans.A.i.m\\
     %  &\to recv.A.i.sig\\
       &\to STOP.\\
 \end{split}&
    \end{align*}
    \end{small}
    \end{spacing}
 \begin{spacing}{0.6}
   \begin{small}
   \begin{align*}
    \begin{split}
USER_{REC}^{I}&=recv.i.A.id_A,id_i,REQU\\
          &\to trans.i.A.id_i,id_A,AGREE\\
          &\to recv.i.A.id_A,id_i,AUTH\\
          &\to trans.i.A.id_i,id_A,auth(id_i,id_A)\\
          &\to recv.i.A.m\\
          &\to STOP.\\
 \end{split}&
    \end{align*}
    \end{small}
    \end{spacing}
 \begin{spacing}{0.6}
   \begin{small}
   \begin{align*}
    \begin{split}
USER_I^{I}(J)&=\Box_{J \in C}trans.i.j.id_i,id_j,k_ia_i\\
      &\Box \quad trans.i.j.id_i,id_j,a_iG\\
      &\to STOP.\\
USER_R^{I}(J)&=\Box_{J \in C}recv.i.i.id_j,id_i,k_ja_j\\
      &\Box \quad trans.i.j.id_j,id_i,a_jG\\
      &\to trans.i.A.sig\\
      &\to STOP.\\
    \end{split}&
    \end{align*}
    \end{small}
    \end{spacing}

The process $USER_{REC}^{I}$ and $USER_R^{I}(J)$ are similar to $REQ_A(I)$ and $USER_I^{I}(J)$. So the entire network is thus modeled as follows.
 \begin{spacing}{0.6}
   \begin{small}
   \begin{align*}
    \begin{split}
    USER^{I}&=\Box_{J \in C}(|||USER_I^{I}(J)) ||| (|||USER_R^{I}(J)),\\
    USER&=|||_{I \in C}USER^{I},USER_{REC}=|||_{I \in C}(USER_{REC}^I),\\
    REQ_A&=|||_{I \in C}REQ_A(I).\\
 NET&=((REQ_A ||| USER_{REC})|||USER) \parallel ENEMY
    \end{split}&
    \end{align*}
    \end{small}
    \end{spacing}
The objective of the scheme  is : $R=\{trans.i.A.id_i,id_A,auth(id_i,id_A)\}$ precedes $T=\{trans.A.i.m\}$.

\subsubsection{Overall analysis }
From the CSP model illustrated above about the schemes above, we can see that:
 that the authentication objective  essentially can be reduced to the authentication objective within the authentication framework between the single initiator and responder.
\emph{(i)} All the authentication objectives in the schemes has the essence that $R$ precedes $T$, here $R$ can be summarized as the  multi-run execution of the authentication framework $auth(\cdot)$. In detail, in the  reconstruction phase and distributed ECC-DSS scheme $R=\{trans.U.U,auth(\cdot)\}$  represents the  maintain unidirectional authentication property. While   the bi-directional authentication property holds in distribution   phase $R$ . And the  $T$ represents the corresponding service and action in different schemes.
\emph{(ii)} The security of  multi-run execution of the authentication framework can be reduced to the single small authentication framework. The detailed proof can be seen in section~\ref{abritrary}.  Take distribution phase scheme as example, the multi-run with arbitrary size authentication component, i.e $ USER(i)$,$USER$, $T$, can be reduced to simple small framework instance. The $USER$ component appears in the complete model of reconstruction and ECC-DSS model.

According to analysis from two aspects above, we can conclude that the security proof towards the authentication property for the small authentication framework can hold the authentication property of complicate multi-run and arbitrary case. So the authentication property of schemes can be maintained well above.

\section{Conclusion}
\label{Conclusion}
In this paper, we proposed a secret key distributed storage scheme based on  CRT-VSS and trusted computing, solving the two category cheating problems in the VSS area before.
Then we explore and do overall analysis of the $(\oplus,\oplus)$,$(\otimes,\otimes)$ homomorphic property with the CRT-VSS scheme and design the corresponding secure shares-product sharing scheme base on CRT scheme.
At last, on the foundation of the robust DSS scheme, our proposed  ECC-DSS distributed authentication scheme based on trusted computing  eliminates the possibility of the malicious attack, e.g. DoS attack and fault attack, during the process of signature generation. The security analysis  proves that our schemes can maintain relative security assurance under certain conditions.

%In the future, we will focus on several aspects as follows: 1) the further formal security analysis,proof and discussion about the schemes proposed should be done by FDR or rank function; 2) the measurement of the configuration of the trusted node should be further studied and refined; 3) the trusted choice relying on the realtime DAA and PBA challenge protocol may be time consuming, some efficiency raise measures can be further discussed, e.g. the transmission strategy of the trusted property, etc.

\section*{Acknowledgment}
This work was supported in part by the National Natural Science Foundation of China (No.61170233, No.61232018, No.61272472, No.61202404, No.61272317)
and China Postdoctoral Science Foundation (No.2011M501060).

%The work presented in this paper was supported by ......

%\renewcommand\refname{Reference}
\bibliographystyle{plain}
\bibliography{ref}

\end{document}